\documentclass[11pt]{article}
\usepackage{times}

\usepackage[left=1.0in, right=1.0in, top=1.0in]{geometry}
\usepackage{xspace}
\usepackage{graphicx}
\usepackage{algorithmic}
\usepackage{algorithm}
\usepackage{amssymb}
\usepackage{amsmath}
\usepackage{amsthm}
\usepackage{algorithm}
\usepackage{algorithmic}
\usepackage{multirow}

\usepackage[linktocpage=true]{hyperref}

\newtheorem{theorem}{Theorem}
\newtheorem{lemma}[theorem]{Lemma}

\newtheorem{corollary}[theorem]{Corollary}

\newtheorem{definition}[theorem]{Definition}

\newcommand{\polylog}{\mathrm{polylog}}
\newcommand{\calP}{\mathcal{P}}

\newcommand{\distP}{\mathbb{P}}

\newcommand{\dst}{\mbox{DST}\xspace}
\newcommand{\kdst}{\mbox{$k$-DST}\xspace}
\newcommand{\tdst}{\mbox{$2$-DST}\xspace}
\newcommand{\kdss}{\mbox{$k$-DSS}\xspace}
\newcommand{\tdss}{\mbox{$2$-DSS}\xspace}
\newcommand{\gst}{\mbox{GST}\xspace}
\newcommand{\kgst}{\mbox{$k$-GST}\xspace}
\newcommand{\kgstx}{\mbox{$k$-GST*}\xspace}
\newcommand{\tgst}{\mbox{$2$-GST}\xspace}
\newcommand{\tgstx}{\mbox{$2$-GST*}\xspace}

\newcommand{\hE}{\hat{E}}
\newcommand{\hV}{\hat{V}}
\newcommand{\hT}{\hat{T}}
\newcommand{\hH}{\hat{H}}
\newcommand{\hS}{\hat{S}}
\newcommand{\hf}{\hat{f}}
\newcommand{\hx}{\hat{x}}
\newcommand{\he}{\hat{e}}
\newcommand{\hv}{\hat{v}}
\newcommand{\hu}{\hat{u}}
\newcommand{\hr}{\hat{r}}

\renewcommand{\setminus}{-}

\newcommand{\eps}{\varepsilon}

\usepackage{color}

\usepackage{float}
\floatstyle{ruled}
\newfloat{Figure}{thp}{lop}
\floatname{Figure}{Fig.}

\ifdefined\DEBUG

 \newcommand{\bun}[1]{\textcolor{blue}{#1}}

  \def\rem#1{{\marginpar{\raggedright\scriptsize #1}}}
   \newcommand{\fabr}[1]{\rem{\textcolor{red}{$\bullet$ #1}}}
   \newcommand{\bunr}[1]{\rem{\textcolor{blue}{$\bullet$ #1}}}

\else

  \newcommand{\bun}[1]{#1}
  \newcommand{\fabr}[1]{}
  \newcommand{\bunr}[1]{}
\fi

\begin{document}

\title{Surviving in Directed Graphs:\\ A Polylogarithmic Approximation for\\ 
       Two-Connected Directed Steiner Tree}

\date{\today}
\author{
Fabrizio Grandoni\thanks{
IDSIA, University of Lugano,  Switzerland, {\tt fabrizio@idsia.ch}. Partially supported by the ERC Starting Grant NEWNET 279352 and the SNSF Grant APPROXNET 200021\_159697/1.}
\and
Bundit Laekhanukit\thanks{
Weizmann Institute of Science, Israel, {\tt bundit.laekhanukit@weizmann.ac.il}. Partially supported by the ISF (grant No. 621/12) and by the I-CORE Program (grant No. 4/11).}
}

\maketitle
\thispagestyle{empty}

\begin{abstract}
Real-word networks are often prone to failures.
A reliable network needs to cope with this situation and must
provide a backup communication channel.
This motivates the study of \emph{survivable network design},
which has been a focus of research for a few decades.
To date, survivable network design problems on undirected graphs 
are well-understood. For example, there is a $2$ approximation 
in the case of edge failures [Jain, FOCS'98/Combinatorica'01].
% and a decent $2$-approximation algorithms
%have been devised [Jain, FOCS'98/Combinatorica'01].
%
The problems on {\em directed graphs}, in contrast, 
have seen very little progress.
Most techniques for the undirected case like 
primal-dual and iterative rounding methods do not seem
to extend to the directed case.
Almost no non-trivial approximation algorithm is known
even for a simple case where we wish to design a network
that tolerates a single failure.

In this paper, we study a survivable network design problem on 
directed graphs, 2-Connected Directed 
Steiner Tree (\tdst): given an $n$-vertex weighted directed graph, a
root $r$, and a set of $h$ terminals $S$,
find a min-cost subgraph $H$ that has
two edge/vertex disjoint paths from $r$ to any $t\in S$.
\tdst is a natural generalization of the classical   
Directed Steiner Tree problem (\dst), where we have an additional
requirement that the network must tolerate one failure.
No non-trivial approximation is known for \tdst.
This was left as an open problem by Feldman~et~al., [SODA'09; JCSS]
and has then been studied by Cheriyan~et~al. 
[SODA'12; TALG] and Laekhanukit [SODA'14]. However, 
no positive result was known except for the special case of
a $D$-shallow instance [Laekhanukit, ICALP'16].

We present an $O(D^3\log D\cdot h^{2/D}\cdot \log n)$ approximation algorithm for \tdst that runs in time $O(n^{O(D)})$, for any $D\in[\log_2h]$. This implies a polynomial-time $O(h^\eps \log n)$ approximation for any constant $\eps>0$, and a poly-logarithmic approximation running in quasi-polynomial time. We remark that this is essentially the best-known even for the classical \dst, and the latter problem is $O(\log^{2-\eps}n)$-hard to approximate [Halperin and Krauthgamer, STOC'03]. As a by product, we obtain an algorithm with the same
approximation guarantee for the $2$-Connected Directed Steiner Subgraph problem, where the goal is to find a min-cost subgraph such that every pair of terminals are $2$-edge/vertex connected. 
%% We can also approximately satisfy side constraints such as degree bounds on the vertices.

%Our technique can be generalised to provide $2$-edge/vertex connectivity between \fab{a given subset of vertices}. 

Our approximation algorithm is based on a careful combination of several techniques. In more detail, we decompose an optimal solution into two (possibly not edge disjoint) \emph{divergent trees} that induces two edge disjoint paths from the root to any given terminal. These divergent trees are then embedded into a shallow tree by means of Zelikovsky's height reduction theorem. On the latter tree we solve a 2-Connected Group Steiner Tree problem and then map back this solution to the original graph. Crucially, our tree embedding is achieved via a probabilistic mapping guided by an LP: This is the main technical novelty of our approach, and might be useful for future work.
\end{abstract}

\newpage

\clearpage
\setcounter{page}{1}

\section{Introduction}
\label{sec:intro}

Real-world networks are often prone to link or node failures. 
A reliable network needs to cope with this situation and must
provide a backup communication channel.
In mathematical terms, ones wish to design a network that
provides a pre-specified number of edge/vertex disjoint paths
between given pairs of {\em terminals}.
This motivates the study of \emph{survivable network design},
which has been a focus of research for a few decades
\cite{SWK69,FJ81,GGPSTW94,J01}.

To date, the survivable network design problems on undirected graphs
are well-understood, and many powerful techniques have been developed
to solve this class of problems. For example, in the edge failure case, there is a $2$-approximation algorithm by Jain~\cite{J01} for the most general version of the problem, {\em Generalized Steiner Network}.
%culminating in a $2$-approximation algorithm for the most general problems,
%namely {\em Generalized Steiner Network}, by Jain~\cite{J01}.
%
In contrast, there has been very slow progress on survivable network
design in {\em directed graphs}.
Most of the standard techniques like primal-dual and iterative
rounding methods do not seem to extend to the directed case.
Positive results are known only for very restricted cases
(see, e.g., \cite{Dahl93,MT04,Gabow07,L16-kdst}).
In fact, there are almost no positive results 
for survivable network design on directed graphs
in the present of Steiner vertices.

%More specifically, the network should provide (at low cost) connectivity between a given set $\mathcal{P}$ of pairs of \emph{terminal} vertices even when a few edges/vertices fail.
%% More specifically, assuming that there are at most $k-1$ edge/vertex faults, this is guaranteed by the presence of $k$ edge/vertex disjoint paths\footnote{By vertex disjoint paths, we mean internally vertex disjoint, i.e. the paths can share the endpoint vertices.} among the considered pairs (i.e., those pairs should be \emph{$k$ edge/vertex connected}).

%% In the case of undirected graphs, the complexity of survivable network design is well understood both for edge and for vertex connectivity (see related work for more details). 
%% %For \fab{edge connectivity}, the problem is NP-hard and APX-hard already for $k=1$ (i.e., in the non-survivable case) and there exists a $2$ approximation algorithm for any $k$ \cite{...}.  \fab{Concerning vertex connectivity, ....}\fabr{I'm less familiar with vertex connectivity: we should also speak about that}
%% The situation is by far less clear in the case of directed
%% graphs\footnote{Note that here the pairs $(s,t)\in \mathcal{P}$ are
%%   ordered, and the considered disjoint paths must be directed from the
%%   source $s$ to the target $t$.}. Here almost no non-trivial result is
%% known.

In this paper, we focus on arguably one of the simplest survivable
network design problems in directed graphs, namely, 2-Connected Directed
Steiner Tree (\tdst):

%\fab{We let the new root be $r^{out}$ and the new sinks be $t^{in}$ for each $t\in S$. Note that we do not change the number $h$ of terminals and we increase the number of vertices by a factor $2$, thus this reduction does not affect asymptotically our approximation factors.} 
%}:
\begin{definition}\label{def:kdst}
In the $2$-connected Directed Steiner Tree problem (\tdst),
we are given an $n$-vertex directed graph $G=(V,E)$ with edge-costs $\{c_e\}_{e\in E}$, a root vertex $r$ and a set of $h$ terminals $S\subseteq V\setminus\{r\}$. The goal is to find a min-cost subgraph $H$ that has at least $2$ edge disjoint paths from $r$ to each $t\in S$.
\end{definition}

Intuitively, the goal of \tdst is to design a network that can function
in the event of one edge failure (thus, it must provide a backup path).
\tdst is a natural generalization of the classical Directed Steiner
Tree problems (\dst), where only one $r,t$-path for each $t\in S$ is
required to exist in $H$. 
Feldman~et~al.~\cite{FKN12} left approximating \tdst as an open problem
\bun{(see also the earlier work in \cite{Dahl93})}, 
and the problem has later been studied in the work
of Cheriyan~et~al. \cite{CLNV14}
and Laekhanukit \cite{L14,L16-kdst}.
However, there was no known non-trivial approximation algorithm for 
\tdst except for the special case of $D$-shallow instances 
(where the length of any root-to-terminal path in 
the optimal solution is at most $D$\footnote{
A $D$-shallow instance is an instance that has an optimal solution $H$
such that, for every terminal $t$, $H$ has $k$ edge-disjoint $r,t$-paths
in which each path has length at most $D$ (i.e., all the $k$ paths
are short). This imitates the notion of the height of a tree, but it
allows $H$ to contain a directed cycle.})~\cite{L16-kdst}.

Here we define \tdst in terms of edge-connectivity.
The vertex-connectivity version is defined analogously, but
we are asked for vertex-disjoint instead of edge-disjoint paths.
%
%% In directed graphs,
%% these two variants are equivalent in terms of approximability.
\bun{The two variants share the same approximability in directed graphs.}
There is a simple reduction that reduces 
the vertex-connectivity version to edge-connectivity version\footnote{In 
more detail, split each vertex $v$ into $v^{in}$ and $v^{out}$, add
a zero-cost edge $v^{in}v^{out}$, and then re-wire each edge
entering and leaving $v$ to $v^{in}$ and $v^{out}$, respectively. A
source-sink pairs $(s,t)$ is then replaced by the pair
$(s^{out},t^{in})$. The number of pairs does not change, and the
number of vertices grows by a factor $2$.} and vice versa. 
We will therefore focus only on the edge connectivity \bun{case}.

\subsection{Our Results and Techniques}

The main contribution of this paper is a non-trivial approximation algorithm for \tdst. 
\begin{theorem} \label{thm:main-2dst}
For any $D\in[\log_2{h}]$, there exists a randomized $O(D^3\log D\cdot h^{2/D}\cdot \log n)$ approximation algorithm for \tdst that runs in $n^{O(D)}$ time.
\end{theorem}
In particular, Theorem~\ref{thm:main-2dst} implies a polynomial-time $O(h^\eps \log n)$ approximation for any constant $\eps>0$, and a quasi-polynomial-time  
$O(\log n\,\log^3 h\log\log h)$ approximation algorithm. We remark that, up to poly-logarithmic factors, this is the best known even for the simpler case of \dst \cite{CCCDGGL99}.

Approximation algorithms for \tdst can be used to approximate with the same asymptotic approximation factor the more general problem, namely {\em \tdss}, described in \cite{CV07,L15-sskcon,N12-sskcon} (see Appendix \ref{sec:fromDSTtoDSS} for more details). 
\begin{definition}\label{def:kdss}
In the $2$-Connected Directed Steiner Subgraph problem (2-DSS),
we are given a directed graph $G=(V,E)$ with edge-costs $\{c_e\}_{e\in E}$ and a set of terminals $S\subseteq V$. The goal is to find a min-cost subgraph $H$ of $G$ such that, for every pair of vertices $s,t\in S$, $H$ has $2$ edge-disjoint paths from $s$ to $t$.
\end{definition}
As a corollary of Theorem~\ref{thm:main-2dst}, we obtain the following result.
\begin{corollary} \label{cor:main-2dss}
For any $D\in[\log_2h]$, there exists a randomized $O(D^3\log D\cdot h^{2/D}\cdot \log n)$ approximation algorithm for 2-DSS that runs in $n^{O(D)}$ time.
\end{corollary}

%\fabr{Still not completely happy with the intuition here: we might add more}
%One can easily reduce the vertex-connectivity case to the edge connectivity case as follows: split each vertex $v$ into $v^{in}$ and $v^{out}$, add a zero-cost edge $v^{in}v^{out}$, and then re-wire each edge entering and leaving $v$ to
%$v^{in}$ and $v^{out}$, respectively. \fab{We let the new root be $r^{out}$ and the new sinks be $t^{in}$ for each $t\in S$.}\footnote{\fab{Note that we do not change the number $h$ of terminals and we increase the number of vertices by a factor $2$, thus this reduction does not affect asymptotically our approximation factors.}} Therefore we will next focus on the \fab{edge} connectivity version only.

Our approach is rather sophisticated, and involves several logical
steps.
The starting point is the following decomposition theorem\footnote{There also exists a vertex-connectivity analogue of this theorem, but we omit it here since it is not necessary for our goals.}.
\begin{theorem}[Divergent Steiner Trees Theorem \cite{GeorgiadisT16,Kovacs07}]
\label{thm:independent-tree}
Let $H$ be a feasible solution to a \tdst instance with a root $r$ and terminals $S$. Then $H$ can be decomposed into two (possibly overlapping) arborescences (\emph{divergent Steiner trees}) $T_1$ and $T_2$ rooted at $r$ and spanning 
$S$ such that, for every terminal $t\in S$, the unique $r$-$t$ paths $P_1$ in $T_1$ and $P_2$ in $T_2$ are edge disjoint.
\end{theorem}
Intuitively, $T_1$ and $T_2$ are two solutions to the \dst problem on
the same instance with the extra property of being edge disjoint
``from the point of view'' of a single terminal. We remark that this
is the only part of our approach that does not directly generalize to
connectivity $k\geq 3$ because the decomposition does not exist for 
$k\geq 3$ \cite{Huck95,BK11}.
(See the discussion in Section~\ref{sec:conclusions}.)

The second main tool from the literature that we wish to exploit is the Zelikovsky's height-reduction theorem \cite{HRZ01,Z97} that is used in approximating \dst.
%Here we have to face the second major problem: we do not know how to (cheaply) round this type of fractional solutions (even for the simpler \dst case) without using Zelikovsky's height-reduction theorem \cite{Z97,HRZ01}.
\begin{theorem}[Height Reduction Theorem \cite{HRZ01,Z97}]
\label{thm:height-reduction-second}
Consider an edge weighted arborescence $T$ rooted at $r$ and spanning $S$.  
Then, for any $D\in[\log_2|S|]$, in the metric completion of $T$, there exists an arborescence $T^D$ of depth at most $D$ rooted at $r$ and spanning $S$
together with a mapping $\psi$ that maps each vertex of $T^D$ to
a vertex of $T$ and a mapping $\phi$ that maps each edge 
$\he=\hu\hv\in E(T^{D})$ to a $\psi(\hu),\psi(\hv)$ path $\phi(\he)$ in $T$ so that  the following \emph{bounded congestion} property holds:
\[
|\he\in E(T^D): e\in \phi(\he))\}| \leq \beta'=O(D\cdot |S|^{1/D})\quad \forall e\in E(T).
\]
In particular, the cost of $T^D$ is at most $\beta'$ times the cost of $T$.
%Moreover, $\phi$ maps a path in $H^D$ to a path in $H$.
\end{theorem}

We remark that the Height Reduction Theorem was originally stated in
terms of cost (\bun{which is implied by our version}). Here we extract the bounded congestion property that is implicit in the proof.
%and \fab{state} it as a theorem.

%
The main difficulty that we have to face is how to apply these
two tools. In DST approximation, Theorem~\ref{thm:height-reduction-second}
is typically applied by considering
the metric closure of the input graph. This is not applicable to
our case since the metric closure might destroy the connectivity
properties of the input graph.
Moreover, we cannot directly apply the theorem to the divergent
Steiner trees because they are a decomposition of 
{\em an optimal solution} that we wish to compute.

We solve these issues by defining an ILP that mimics the decomposition
of the optimal solution into divergent Steiner trees $T_1$ and $T_2$
(as in Theorem~\ref{thm:independent-tree}) and the following application of 
Theorem~\ref{thm:height-reduction-second} to these trees to obtain
$D$-shallow trees $T_1^D$ and $T_2^D$.
In more detail, we define a $D$-shallow tree that incorporates (twice)
all the possible paths of length $D$ starting from the root
(analogously to \cite{L16-kdst}).
This shallow tree implicitly includes $T_1^D$ and $T_2^D$.
We encode the mapping of each edge of $T_1^D\cup T_2^D$ 
into the associated paths in $T_1\cup T_2$ using flow constraints.
We also add constraints that encode the \emph{bounded congestion} property from Theorem~\ref{thm:height-reduction-second} (crucial to bound the cost of the approximate solution) and the \emph{divergency} property from Theorem \ref{thm:independent-tree} (crucial to achieve a feasible solution).

Rounding a fractional solution to the linear relaxation is a non-trivial task. 
We observe that each terminal $t\in S$ is associated with a subset of
vertices $\hS_t$ in the shallow tree, and the edges of $T_1^D\cup
T_2^D$ must contain two edge disjoint paths from the root $\hr$ to
$\hS_t$.
In other words, the latter edges induce a feasible solution to a tree
instance of $2$-Edge Connected Group Steiner Tree (\tgst) with root $\hr$
and groups $\{\hS_t\}_{t\in S}$ (more details in related work). 
This allows us to add the standard LP constraints for \tgst on 
{\em a tree} to our linear relaxation, and use the GKR rounding algorithm 
by Garg~et~al. \cite{GKR00} to round the corresponding variables to an
integral \tgst solution in the shallow tree.

The last obstacle that we need to face is that we need to map back each chosen edge $\he$ of the shallow tree to a path $\phi(\he)$ of the original graph. The LP solution provides a \emph{fractional mapping} in the form of a flow. We interpret this flow as a distribution over paths and sample one path $\phi(\he)$ according to this distribution. In order to show that the solution is feasible (with large enough probability), we exploit an argument similar in spirit to the one used by Chalermsook et al. \cite{CGL15} in the framework of $k$-Edge Connected Group Steiner  Tree (\kgst) approximation. However, our probabilistic mapping makes the analysis slightly more involving. Shortly, we argue that for any given edge $e$ of the original graph, GKR rounding has sufficiently large probability to select paths using only edges $\he$ of the shallow tree whose associated probabilistic mapping has low chance to use the edge $e$. The claim then follows by a cut argument as in \cite{CGL15}.

\subsection{Related Work}

In the Directed Steiner Tree problem (\dst), we are given an $n$-vertex directed edge weighted graph, a root $r$ and a collection of $h$ terminal vertices $S$. The goal is to find a min-cost arborescence rooted at $r$ and spanning $S$. \dst is one of the most fundamental network design problems in directed graphs. 
\dst admits, for any positive integer $D$, an $O(D h^{1/D}\log^2h)$ approximation running in time $n^{O(D)}$ ~\cite{CCCDGGL99,Z97}. In particular, this implies a polynomial-time $O(h^\eps)$ approximation for any constant $\eps>0$, and an $O(\log^3 h)$ approximation running in quasi-polynomial time. 

\kdst\ and \kdss are the natural generalization of \tdst and \tdss,
respectively, with connectivity $k$. 
These problems have been a subject of study since early 90's \cite{Dahl93}
and have been subsequently studied in \cite{CLNV14,L14,L16-kdst}.
Cheriyan~et~al. \cite{CLNV14} showed that \kdst is at least as hard as 
the {\em Directed Steiner Forest} problem
and the {\em Label-Cover} problem.
Thus, \kdst admits no $2^{\log^{1-\eps}n}$-approximation
algorithm, for any $\eps>0$,
unless $\mathrm{NP}\subseteq\mathrm{DTIME}(2^{\polylog(n)})$.
For small $k$, they showed that \kdst admits no
$k^{\sigma}$-approximation algorithm for some fixed $\sigma>0$
unless $\mathrm{P}=\mathrm{NP}$.
If $k$ is large enough, then \kdst is NP-hard even when
we have only two terminals,
and they further proved that \kdst when $h$ and $k$ are constants is 
polynomial-time solvable in directed acyclic graphs.
However, if the input graph contains a cycle, 
the complexity status of \kdst is not clear even for $k=h=2$.
Laekhanukit refined the hardness result of \kdst in
\cite{L14}
and showed that \kdst admits no
$k^{1/2-\eps}$-approximation algorithm, for any
constant $\eps>0$, unless $\mathrm{NP}=\mathrm{ZPP}$.
In a subsequent work, Laekhanukit \cite{L16-kdst} presented an LP-based $O(k^{D-1}D\log n)$-approximation
algorithm for {\em $D$-shallow} instances of \kdst and \kdss running in time $n^{O(D)}$. It seems that his approach cannot be generalized to arbitrary instances (although we will exploit part of his ideas).

A well-studied special case of \dst is the \emph{Group Steiner Tree} problem (\gst). Here we are given an undirected weighted graph, a root vertex $r$, and a collection of $h$ groups $S_i\subseteq V$. The goal is to compute the cheapest tree that spans $r$ and at least one vertex from each group $S_i$. The best-known polynomial-time approximation factor for \gst\ is
$O(\log^2h\log n)$ due to Garg~et~al.~\cite{GKR00}. Their algorithm uses {\em  probabilistic distance-based tree embeddings}
\cite{B96,FRT04} as a subroutine.
Chekuri and Pal \cite{CP05} presented an $O(\log^2 h)$ approximation
that runs in quasi-polynomial time.
On the negative side,
Halperin and Krauthgamer \cite{HK03}
 showed that \gst\ admits 
no $\log^{2-\eps}n$-approximation algorithm, for
any constant $\eps>0$, unless 
$\mathrm{NP}\subseteq\mathrm{ZPTIME}(2^{\polylog(n)})$. This implies the same hardness for \dst, hence for \tdst and \tdss.

The high-connectivity version of \gst, namely, 
the {\em $k$-Edge Connected Group Steiner Tree problem} (\kgst),
was studied in \cite{CGL15,GKR10,KKN12}. Here the goal is to find a min-cost subgraph that contains $k$ edge-disjoint paths between the root and each group. For $k=2$, 
the best approximation ratio is $\tilde{O}(\log^3n\log h)$
due to the work of Gupta~et~al.~\cite{GKR10}.
If the size of any group is bounded by $\alpha$, then
there is an $O(\alpha \log^2n)$-approximation algorithm by
Khandekar~et~al.~\cite{KKN12}.
For $k\geq 3$, there is no known non-trivial approximation algorithm
for \kgst. Chalermsook~et~al. \cite{CGL15} presented
an LP-rounding bicriteria approximation algorithm that returns a
subgraph with cost $O(\log^2n\log h)$ times the optimum 
while guarantees a connectivity of at least $\Omega(k/\log n)$. Their algorithm uses the {\em probabilistic cut-based tree embeddings} by
R{\"{a}}cke~\cite{R08} as a subroutine (as opposed to distance-based ones in \cite{GKR10}). We will exploit part of their ideas in our rounding algorithm 
(although \bun{a probabilistic tree embedding for directed graphs
is not available for us}).
Chalermsook~et~al. also showed that \kgst\ is hard to approximate to
within a factor of $k^{\sigma}$, for some fixed constant $\sigma>0$,
and if $k$ is large enough, then the problem is at least as hard as
the Label-Cover problem, meaning that \kgst\ admits no
$2^{\log^{1-\eps}n}$-approximation algorithm, for any constant $\eps>0$,
unless $\mathrm{NP}\subseteq\mathrm{DTIME}(2^{\polylog(n)})$.

As already mentioned, survivable network design is well studied in undirected (weighted) graphs. First, consider the edge connectivity version. The earliest work is initiated in early 80's by Frederickson and J{\'{a}}J{\'{a}}~\cite{FJ81}, where the authors studied the {\em 2-Edge Connected Subgraph} problem in both directed and undirected graphs.
In the most general form of the problem, also known as the \emph{Steiner Network} problem, we are given non-negative integer requirements $k_{u,v}$ for all pairs of vertices $u,v$, and the goal is to 
find a min-cost subgraph that has $k_{u,v}$ edge-disjoint paths between $u$ and $v$. Jain \cite{J01} devised a $2$-approximation algorithm for this problem. We remark that $2$ is the best known approximation factor even for $k_{u,v}\in \{0,1\}$ \cite{AKR95}, which is known as the \emph{Steiner forest} problem. The classical \emph{Steiner tree} problem is a special case of Steiner forest where all pairs share one vertex. Here the best known approximation factor is $1.39$ due to the work of
Byrka~et~al.~\cite{BGRS13}.

Concerning vertex connectivity, two of the most well-studied problems are the {\em $k$-Vertex Connected Steiner Tree} ($k$-ST) and {\em $k$-Vertex Connected Steiner Subgraph} ($k$-SS) problems, i.e., the undirected versions of \kdst\ and \kdss, respectively.  There are $2$-approximation algorithms for $2$-ST and $2$-SS by Fleischer~et~al. \cite{FJW06} using the iterative rounding method.
For $k\geq 3$, Nutov devised an $O(k\log k)$-approximation algorithm
for $k$-ST in \cite{N12-rooted} and 
an $O(\min\{|S|^2, k\log^2 k\})$-approximation algorithm for $k$-SS in
\cite{N12-sskcon} (also, see \cite{L15-sskcon}).
A special case of $k$-SS with metric-costs is
studied by Cheriyan and Vetta in \cite{CV07} who
gave an $O(1)$-approximation algorithm for the problem.
The most extensively studied special case of $k$-SS is when all
vertices are terminals, namely the {\em $k$-Vertex Connected Spanning Subgraph} problem, which has been studied, e.g., in \cite{CVV03,KN05,FL12,N14-kvcss,CV14}.
The current best approximation guarantees are
$O(\log(n/(n-k))\log{k})$ \cite{N14-kvcss}, and
$6$ for the case $n \leq 2k^3$ \cite{CV14,FNR15}. More references can be found in \cite{KN07,N16-survey-Steiner,N16-survey-Spanning}.

\paragraph{Notation.}
%\label{sec:prelim}
We use standard graph terminologies. For any graph $G$, we denote vertex and edge sets of $G$ by 
$V(G)$ and $E(G)$, respectively. For any subset of vertices $S\subseteq V(G)$ (or a single vertex $S=v$), 
we denote the set of edges of $G$ entering $S$ by $\delta^{in}_G(S)$
and denote the set of edges leaving  $S$ by $\delta^{out}_G(S)$.
%In other formally,
%\[
%\delta^{in}_G(S)=\{uv\in E(G):v\in S \land u\not\in S\}
%\mbox{ and }
%\delta^{out}_G(S)=\{uv\in E(G):u\in S\land v\not\in S\}.
%\]

\section{Embedding into a Shallow Tree}
\label{sec:shallow}

Our LP-relaxation is defined based on the existence an embedding of an
optimal \tdst solution $H$ in the original graph into an auxiliary
$D$-shallow tree $\hat{H}$ (i.e., a tree of depth at most $D$), where
$D>0$ is an integer given as parameter. 
Our embedding is obtained by applying 
the \emph{Height Reduction} to \emph{Divergent Steiner Trees}.

We start by decomposing $\hat{H}$ into two divergent Steiner trees
$T_1$ and $T_2$ using the Divergent Steiner Tree Theorem
(Theorem~\ref{thm:independent-tree}). 
Then we apply the Height Reduction Theorem
(Theorem~\ref{thm:height-reduction-second}) to each such $T_i$, hence
getting a $D$-shallow tree $T^D_i$ in the metric closure of $T_i$
together with mappings $\psi_i$ and $\phi_i$.
The final step is to unify the roots of $T^D_1$ and $T^D_2$, hence
getting a tree $\hat{H}$ rooted at $\hat{r}$. 
We also merge the two mappings in a natural way, thus 
getting $\psi:V(\hH)\rightarrow V(H)$ and $\phi:E(\hH)\rightarrow
2^{E(H)}$. 
Let $\psi^{-1}(v)$ be the set of vertices $\hat{v}\in V(\hat{H})$ with
$\psi(\hat{v})=v$. 
Note also that each simple $\hu,\hv$-path $\hat{P}$ in $\hH$ defines a
$\psi(\hu),\psi(\hv)$ path $P=\phi(\hat{P})$ in $H$.

By construction, it is not hard to see that $(\hH,\psi,\phi)$ has the
following properties: 
\begin{enumerate}\itemsep0pt
\item (\emph{divergency}) for any terminal $t\in S$, there exist two vertices $\hat{t}_1,\hat{t}_2\in \psi^{-1}(t)$ such that the following holds. Let $\hat{P}_i$ be the $\hat{r}$-$\hat{t}_i$ path in $\hat{H}$ for $i=1,2$. Then $\phi(\hat{P}_1)$ and $\phi(\hat{P}_2)$ are two edge-disjoint $r$-$t$ paths in $H$ (and consequently also in $\hat{H}$). 
\item (\emph{bounded congestion}) For any edge $e\in E(H)$, $|\hat{e}\in E(\hat{H}): e\in \phi(\hat{e})|\leq \beta:=2\beta'=O(D|S|^{1/D})$.
\end{enumerate}

Note that we do not know an optimal solution, and consequently the two trees $T_1,T_2$ that are needed to define the above embedding. In the next section, we define an LP relaxation that, in some fractional sense, achieves this goal. 

\section{An LP-relaxation for \tdst}
\label{sec:LP}

In this section, we present an ILP formulation of \tdst, and the corresponding LP relaxation.

The first step in the definition of our ILP is to build a proper
$D$-shallow tree $\hat{T}=(\hV,\hE)$ that contains the tree $\hat{H}$
(defined in the previous section) as a subgraph.
To this end, we list {\em twice} all the possible sequences of at most $D+1$
distinct vertices of $G$ starting with the root $r$. 
The {\em prefix tree} of these sequences (rooted at $\hat{r}=r$)
is our tree $\hat{T}$. That is, each vertex $\hat{v}$ of $\hat{T}$ is associated with a
vertex $v$ in the input graph $G$, 
and each rooted-path in $\hat{T}$ corresponds to each sequence we
listed.
It is not hard to see that $\hat{H}$ can be mapped to a subtree of $\hat{T}$. 
Let $\psi:\hV\rightarrow V$ be the corresponding mapping of vertices. With the same notation as before, we define 
$\hat{S}_t:=\psi^{-1}(t)$ to be the set of vertices in $\hT$ corresponding to terminal $t\in S$ (the \emph{group} of $t$). The notion of group will be needed later to define a proper \tgst instance.

%Let $\psi^{-1}(v):=\{\hat{v}\in V(\hat{T}): \psi(\hat{v})=v\}$ be the vertices of $\hat{T}$ which are mapped into $v\in V(G)$. For a terminal $t\in S$, we let $\hat{S}_t:=\psi^{-1}(t)$ be the corresponding \emph{group}. The notion of group will be needed later to define a proper $2$-GST instance.

We have all the ingredients for formulating our ILP.
We define indicator variables $x_e\in \{0,1\}$ for all $e\in E$, which take value $x_e=1$ iff $e\in H$ ($H$ is an optimal solution to \tdst). 
The objective function that we wish to minimize is $\sum_{e\in E}c_ex_e$.
Similarly, we define indicator variables $\hat{x}_e\in \{0,1\}$ for all $e\in \hE$, which take value $\hat{x}_e=1$ iff $e\in E(\hH)$.

Now we define our constraints.
First we define a set of linear constraints, denoted by $LP_{gst}$, which models the fact that, for each $t\in S$, $\hH$ must contain two edge disjoint paths from $\hat{r}$ to the group $\hS_t$. So, we introduce flow variables $\hf^t_{\he}\in \{0,1\}$ for all $\he\in \hE$ and all terminals $t\in S$. The constraints $LP_{gst}$ are given in Figure~\ref{fig:LP2GST}.
\begin{Figure}
\everymath{\displaystyle}
\[
\begin{array}{rclll}
%  min & $\sum_{e\in E(G)}c_ex_e$\\
%  s.t.
     \hf^t_{\he}  &\leq &  \hx_{\he}
      &\quad& \forall \he\in \hE,\forall t\in S\\
     \sum_{\he\in\delta^{in}_{\hT}(\hv)}\hf^t_{\he} & = &  
      \sum_{\he\in\delta^{out}_{\hT}(\hv)}\hf^t_{\he} 
      &\quad& 
      \forall t\in S,\forall \hat{v}\in \hV\setminus (\hat{S}_t\cup \{\hat{r}\})\\
     \sum_{\he\in\delta^{in}_{\hT}(\hat{S}_t)}\hf^t_{\he} & \geq &  2
      &\quad& \forall t\in S
\end{array}
\]
\everymath{\textstyle}
\caption{The $LP_{gst}$ constraints.} 
\label{fig:LP2GST}
\end{Figure}
We remark that $LP_{gst}$ are the linear constraints of the standard LP relaxation for the \tgst problem with the root $\hat{r}$ and groups $\hat{S}_t$ for $t\in S$ in which the underlying graph is a {\em tree}. This is a crucial part of our formulation because this LP has a large integrality gap on general graphs \cite{ZosinK02}.

Next we define the set of constraints $LP_{cong}$ that formulates
(implicitly) a mapping $\phi:\hE \rightarrow 2^{E}$ of edges
$\he=\hu\hv$ of $\hT$ into $\psi(\hu),\psi(\hv)$ paths of $G$. We
introduce the following new flow variables: $f_{\he,e}\in \{0,1\}$,
for all $\he\in \hE$ and $e\in E$. Intuitively, the set of edges $e\in E$
with $f_{\he,e}=1$ form the path $\phi(\he)$. 
Clearly one has $f_{\he,e}\leq x_e$.
In order to satisfy the bounded congestion property, 
we impose that, for a given $e\in E$, the sum of variables $f_{\he,e}$ 
is upper bounded by $\beta\cdot x_e$, where $\beta=O(D|S|^{1/D})$ 
comes from the Height Reduction Theorem 
(Theorem~\ref{thm:height-reduction-second}). These LP constraints are given in Figure~\ref{fig:LP-cong}.
\begin{Figure}
\everymath{\displaystyle}
\begin{center}
\begin{tabular}{rllll}
$f_{\he,e}$
  & $\leq$ 
  & $x_e$
  & $\forall \he=\hu\hv \in \hE,\forall e \in E$\\ 
$\sum_{e\in\delta^{out}_G(u),u=\psi(\hu)}f_{\he,e}$
  & $=$ 
  & $\hx_{\he}$
  & $\forall \he=\hu\hv \in \hE$ \\
$\sum_{e\in\delta^{in}_G(u),u=\psi(\hu)}f_{\he,e}$
  & $=$ 
  & $0$ 
  &$\forall \he=\hu\hv \in \hE$\\ 
$\sum_{e\in\delta^{in}_{G}(w)}f_{\he,e}$  
  & $=$ 
  & $\sum_{e\in\delta^{out}_G(w)}f_{\he,e}$ 
  & $\forall \he=\hu\hv \in \hE, \forall w \in V\setminus \{\psi(\hu),\psi(\hv)\}$ \\
$\sum_{\he\in \hE}f_{\he,e}$  
  & $\leq$ 
  & $\beta \cdot x_e$ 
  & $\forall e\in E$
\end{tabular}
\end{center}
\caption{The constraints $LP_{cong}$.} 
\label{fig:LP-cong}
\end{Figure}

It remains to enforce the divergency property.
We introduce a final set of new variables: 
$f^{t}_{\he,e}\in \{0,1\}$, for all $\he\in \hE$, $e\in E$, and $t\in
S$. Intuitively, the edges $e\in E$ with $f^t_{\he,e}=1$ indicate
whether $e$ is part of one of the two edge disjoint paths in $H$ from
$r$ to $t$. In an integral solution, for a given $e\in E(H)$ and $t$,
at most one $f^t_{\he,e}$ can be set to $1$. This guarantees that the
mapping $\phi$ maps two $\hat{r}$-$\hS_t$ edge-disjoint paths in the
shallow tree into two edge disjoint paths in the original graph from
$r$ to $t$. The set of constraints $LP_{div}$ is described in
Figure~\ref{fig:LP_div}. 

\begin{Figure}
\everymath{\displaystyle}
\begin{center}
\begin{tabular}{rllll}
$f^t_{\he,e}$               
  & $\leq$ & $f_{\he,e}$
  & $\forall e\in E, \forall \he\in \hE, \forall t\in S$\\
$\sum_{e\in\delta^{out}_G(u),u=\psi(\hu)}f^t_{\he,e}$
  & $=$ & $\hf^t_{\he}$
  & $\forall \he=\hu\hv \in \hE,\forall t\in S$\\
$\sum_{e\in\delta^{in}_G(u),u=\psi(\hu)}f^t_{\he,e}$
  & $=$ & $0$
  & $\forall \he=\hu\hv \in \hE,\forall t\in S$\\
$\sum_{e\in\delta^{in}_G(w)}f^t_{\he,e}$
  & $=$ & $\sum_{e\in\delta^{out}_G(w)}f^t_{\he,e}$
  & $\forall \he=\hu\hv \in \hE,\forall t\in S,\forall w \in V\setminus \{\psi(\hu),\psi(\hv)\}$\\
$\sum_{\he\in \hE}f^t_{\he,e}$               
  & $\leq$ & $x_e$
  & $\forall e\in E, \forall t\in S$
\end{tabular}
\end{center}
\caption{The constraints $LP_{div}$}
\label{fig:LP_div}
\end{Figure}

By relaxing the integrality constraints on the variables, we obtain an
LP relaxation LP-\tdst for \tdst, presented in Figure~\ref{fig:globalLP}. 
\begin{Figure}
\everymath{\displaystyle}
\begin{tabular}{llll}
  min & $\sum_{e\in E}c_e\,x_e$\\
  s.t. & $LP_{gst}$ & \\
        & $LP_{cong}$ & \\ 
        & $LP_{div}$ & \\ 
& $0 \leq x_e, \hx_{\he}, f_{\he,e}, \hf^t_{\he},
    f^t_{\he,e} \leq 1$
& $\forall \he\in \hE, \forall e\in E, \forall t\in S$    
\end{tabular}
\everymath{\textstyle}
\caption{LP relaxation LP-\tdst.} 
\label{fig:globalLP}
\end{Figure}

\section{Approximation Algorithm: Rounding via Tree Embedding}
\label{sec:algo}

In this section, we present our approximation algorithm for \tdst. Our algorithm starts by solving LP-\tdst.
Denote by $\{x_e, \hx_{\he}, \hf^t_{\he},f_{\he,e},f^t_{\he,e}\}_{e\in E,\he\in \hE,t\in S}$ an optimal fractional solution.
We then execute for $O(D \log n)$ times a \emph{rounding procedure}, consisting of two main steps: the \emph{GST rounding} and the \emph{path mapping}. The union of all the solutions obtained is the approximate solution, which is feasible w.h.p.
   
In more detail, consider a given iteration $j$. The variables $\{\hx_{\he}\}_{\he\in \hE}$ provide a feasible solution to the standard LP for \tgst on trees. In the GST rounding step, we apply the rounding algorithm by Garg~et~al.~\cite{GKR00}, which we refer to as GKR rounding, to round these variables. This gives us a subtree $\hat{H}_j=(\hV_j,\hE_j)$ of $\hT$. 

In the path mapping step, we consider each edge $\he=\hu\hv\in \hE_j$, where $u=\psi(\hu)$ and $v=\psi(\hv)$. We randomly map $\he$ to a $u,v$-path in $G$. To this aim, we interpret variables $\{\hf_{\he,e}\}_{e\in E}$ as a distribution $\distP_{\he}$ over $u,v$-paths, and we sample according to this distribution. We repeat this sampling $O(\beta\log D)$ many times to guarantee that we have the desired properties 
(which will be discussed later) with sufficiently large probability.

%
%We apply the entire algorithm for $O(D\log n)$ times to guarantee that  
%we have connectivity two from $r$ to each terminal $t_i\in S$. 

Our main algorithm is presented in Algorithm~\ref{algo:main-algo}.

\begin{algorithm}
\caption{Round \tdst}
\begin{algorithmic}[1]
\label{algo:main-algo}
\STATE Solve LP-\tdst and obtain a fractional solution 
       $\{x_e,\hx_{\he},\hf^t_{\he},
         f_{\he,e},f^t_{\he,e}\}_{e\in E,\he\in \hE,t\in S}$.
\FOR{$j=1$ \TO $20 D\ln n$} \label{algo:tdst:cut-avoiding}
   \STATE Round variables $\hat{x}_{\he}$ using GKR Rounding, and obtain a subtree $\hH_j=(\hV_j,\hE_j)$ of $\hT$.
    \FOR{each $\he=\hu\hv\in \hE_j$, $u=\psi(\hu)$ and $v=\psi(\hv)$}
       \FOR{$\ell=1$ \TO $(4\beta+2)\ln D$} \label{algo:tdst:path-sampling}
           \STATE Sample a $u,v$-path $P^{\he}_{\ell}$ in $G$ from the distribution  $\distP_{\he}$.
                    \ENDFOR
                 \ENDFOR  
        \STATE Let $H_j$ be the union of all sampled paths          
%   \STATE For each edge $e=\hat{u}\hat{v}\in E(\calH_{j})$,
%             map $\hat{a}$ to a $u,v$-path in $G$ by sampling
%             $4\beta\ln D$ paths from a distribution 
%             $\distP_{\hat{a}}$, where $u=\psi(\hat{u})$ and
%             $v=\psi(\hat{v})$.
%   \STATE Denote the mapped solution by $\tilde{H}_j$.
\ENDFOR
\RETURN $H := \cup_{j}H_{j}$.
\end{algorithmic}
\end{algorithm}

The GKR Rounding algorithm is discussed in Section~\ref{sec:GKR-Rounding}.
The construction of path distributions is presented in
Section~\ref{sec:path-distribution}.
We then analyze our algorithm in Section~\ref{sec:analysis}.

\subsection{GKR Rounding}
\label{sec:GKR-Rounding}

Let $\hat{T}$ be the shallow tree. 
We may think that each edge is directed from the root. 
Let $\rho(\he)$ denote a parent of an edge $\he\in
\hE$, i.e., $\rho(\he)$ is an edge adjacent to $\he$ that is
closer to the root. 
Consider the constrains $LP_{gst}$ on variables $\hx_{\he}$ and $\hf^t_{\he}$. This is indeed the standard LP for \tgst. Hence, we can apply GKR rounding. Assume w.l.o.g. that $\hat{x}_{\he} \leq \hat{x}_{\rho(\he)}$.
GKR algorithm considers edges in order of increasing distance from the root. Each edge $\he=\hat{r}\hv$ incident to the
root $\hat{r}$ is marked independently with probability $\hx_{\he}$. Any other edge $\he\in \hE$ whose parent edge has been marked is marked independently with   
probability $\hat{x}_{\he}/\hat{x}_{\rho(\he)}$.
Each marked edge is added to the output tree.
%\bun{if all of its descendants are marked
In our case, this gives the graph $\hat{H}_j$.

Next lemma summarizes the properties of GKR Rounding that we will need in the analysis.

\begin{lemma} [\cite{GKR00,R11}]
\label{lem:property-GKR}
Consider the run of GKR rounding algorithm on a $D$-shallow
tree $\hat{T}$ with variables 
$\{\hx_{\he},\hf^t_{\he}\}_{\he\in \hE,t\in S}$ given by a fractional solution to the standard \gst\ LP. Let $\hat{H}$ be the solution sampled by the algorithm, $t\in S$, and $\mu_t:=\sum_{\he\in\delta^{in}_{\hT}(\hS_t)}\hf^t_{\he}$. Then 
$$
\Pr[\he \in \hH]=\hx_{\he}\quad\mbox{and}\quad
\Pr[\mbox{$\hat{H}$ contains an $\hat{r},\hS_t$-path}]
\geq
  \frac{\mu_t}{2D}.
$$
%Each edge $\he\in \hE$ is chosen with probability
%$\hx_{\he}$. 
%%Thus, $\E[\hat{a}\in E(\hat{H})] = \hat{x}_{\hat{a}}$.
%Moreover, let
%$\mu_i=\sum_{\hat{a}\in\delta^{in}_{\hat{T}}(\hat{S}_i)}\hat{f}^i_{\hat{a}}$ be the
%value of flow that a group $\hat{S}_i$ receives from the root $\hat{r}$.
%Then the probability that $H$ has an $\hat{r},\hat{S}_i$-path is at
%least $1/2D$.
%That is, 
%\[
%\E[\hat{a}\in E(\hat{H})] = \hat{x}_{\hat{a}} \quad\mbox{and}\quad
%\Pr[\mbox{$\hat{H}$ has an $\hat{r},\hat{S}_i$-path}]
%\geq
%  \frac{\sum_{\hat{a}\in\delta^{in}_{\hat{T}}(\hat{S}_i)}\hat{f}^i_{\hat{a}}}{2D}.
%\]
%\qed
\end{lemma}

\subsection{Constructing Path Distributions}
\label{sec:path-distribution}

Now we discuss how to construct a path distribution $\distP_{\he}$
on each edge $\he=\hu\hv\in \hE$.
Let $u=\psi(\hat{u})$ and $v=\psi(\hat{v})$.
Observe that the variables 
$F=\{f_{\he,e}\}_{e\in E}$ form a $u,v$-flow.
Thus, we can decompose $F$ into a collection of flow paths,
say $\{f^{\he}_{P}\}_{P\in \calP_{uv}}$,
where $\calP_{uv}$ is the set of all $u,v$-paths in $G$, so that
\[
\sum_{P\in\calP_{uv}:e\in E(P)}f^{\he}_P = f_{\he,e}.
\]
The value of the flow $F$ is
$\hat{x}_{\he}=\sum_{e\in\delta^{out}_{G}(u)}f_{\he,e}$. 
Thus, $\{f^{\he}_P/\hx_{\he}\}_{P\in\calP_{uv}}$ gives a
collection of flow paths whose total flows is one,
and we can interpret this as a distribution over flow paths.
This implies the following lemma. 

\begin{lemma}
\label{lem:path-distribution}
Consider an edge $\he=\hu\hv\in \hE$ and its corresponding 
variables $\{f_{\he,e}\}_{e\in E}$.
Let $u=\psi(\hat{u})$ and $v=\psi(\hat{v})$.
Then there exists a distribution of $u,v$-paths $\distP_{\he}$
such that for all $e \in E$:
\[
\Pr_{P\sim\distP_{\he}}[e \in E(P)] 
  = \frac{1}{\hat{x}_{\he}}\cdot f_{\he,e}
\]
\end{lemma}

\subsection{Analysis}
\label{sec:analysis}

Next we analyze Algorithm~\ref{algo:main-algo}. We start with the simpler part of our analysis, namely bounding the expected cost of $H$. 
%This is an easy part which follows by a straightforward analysis.
%We show that the expected cost of $H$ is at most 
%$O(\beta^2 D\log D\log n)$ of an optimal solution.
%More precisely, we prove the following lemma.

\begin{lemma}
\label{lem:expected-cost}
The expected cost of $H$ is 
$O(D^3h^{2/D} \log D\log n) \cdot \sum_{e\in E}c_ex_e$.
\end{lemma}

\begin{proof}
Let us bound the expected cost of $H_j=(V_j,E_j)$. For each edge $e\in E$, 
\begin{eqnarray*}
\Pr[e\in E_j]  & \leq &   \sum_{\he\in \hE}
     \Pr\left[\he\in \hE_j\right]\cdot \Pr\left[ 
         \left(e\in \bigcup_{\ell=1}^{(4\beta+2)\ln D}E\left(P^{\he}_{\ell}\right)\right)
     \bigg | \he\in \hE_j \right]\\     
& \overset{\text{Lem. \ref{lem:property-GKR} and~\ref{lem:path-distribution}}}{\leq}  &   \sum_{\he\in \hE}
     \hx_{\he}\cdot O(\beta \log D)\cdot \frac{f_{\he,e}}{\hx_{\he}} \overset{\text{by $LP_{cong}$}}{\leq} O(\beta^2 \log D)\cdot x_e
\end{eqnarray*}
Thus, the expected cost of $H_j$ is $O(\beta^2 \log D)$ times the LP value. 
The claim follows since there are $O(D \log n)$ iterations and $\beta=O(Dh^{1/D})$.
\end{proof}

%\paragraph{Feasibility Analysis.}

We next show that our algorithm gives a feasible solution to \tdst
with high probability. This is the most complicated part of the analysis.
%
%Our analysis exploits some ideas from 
%\cite{CGL15}, but is substantially more complicated due to our probabilistic mapping of edges of the shallow tree into paths of the input graph. 
%However, we need a very delicate analysis because 
%our mapping is ``not deterministic''.

%
%To show this, we first have to show that the shallow tree $\hat{T}$ 
%supports a flow of value at least one for every group $\hat{S}_t$ 
%even if the capacity of edges $\he$ corresponding to $e$ is set to $0$,
%which is the same as setting
%$\hx'_{\he} := \hx_{\he} - x_{\he,e}$.\fabr{This part was not clear to me} 
%

The initial part of our analysis resembles the analysis in \cite{CGL15} for \kgst. We prove feasibility using Menger's theorem and a cut argument. 
By Menger's theorem, the solution subgraph $H\subseteq G$ contains two
edge disjoint $r,t$-paths if and only if 
$H\setminus \{e\}$ contains an $r,t$-path for every edge $e\in E(H)$.
Therefore, we will focus on a given such pair $(e,t)$.
Our goal is to show that our rounding algorithm buys with sufficiently high probability an $r,t$-path \emph{not using} the edge $e$. For this purpose, we exploit the fact that, according to Lemma \ref{lem:property-GKR}, if we reduce the flow associated to some group $\hS_t$, GKR algorithm will still connect $\hr$ to $\hS_t$ with sufficiently large probability provided that the residual amount of flow $\mu'_t$ from $\hr$ to $\hS_t$ is large enough. 

At this point, we might try to reduce the flow
by the amount $f_{\he,e}$ for each edge $\he\in \hE$. One can show that
$\mu'_t$ would remain large enough, but unfortunately this in not
sufficient in our case. Indeed, we might still have a fairly high
probability to use the edge $e$ due to the probabilistic distribution
over paths: for any given $\he$, the sampled path $P^{\he}_\ell$
contains $e$ with probability $f_{\he,e}/\hx_{\he}$. So, we can
``safely'' use the edge $\he$ only if the complementary probability
$(\hx_{\he}-f_{\he,e})/\hx_{\he}$ is sufficiently large. We say that
an edge of the latter type is good, and we wish to use only good edges. 

\bun{Formally}, we say that an edge $\he\in \hE$ is {\em good} against $e$ if
\begin{align*}
%\label{eq:def-good-edges}
\hx_{\he} - f_{\he,e} 
   \geq \frac{1}{2\beta}\cdot f_{\he,e}.
\end{align*}
Otherwise, we say that $\he$ is {\em bad} against $e$.
If the edge $e$ is clear from the context, we will simply
say that $\he$ is bad (respectively, good).
Similarly, we say that a path $\hat{P}$ in $\hat{T}$ is 
{\em good} (against $e$) if all edges of $\hat{P}$ are
good. Otherwise, we say that $\hat{P}$ is bad.

We claim that we can route an $\hat{r},\hS_t$-flow of value at least $1/2$
using only good edges and even after decreasing the capacity of edge $\he$ by $f_{\he,e}$.
In particular, we prove the following lemma.

%Now let us formally define good and bad edges.
%We say that an edge $\hat{a}\in E(\hat{T})$ is {\em good} against $e$ if
%\begin{align}
%\label{eq:def-good-edges}
%\hat{x}_{\hat{a}} - z_{\hat{a},e} 
%   \geq \frac{1}{2\beta}\cdot z_{\hat{a},e}.
%\end{align}
%Otherwise, we say that $\hat{a}$ is {\em bad} against $e$.
%%
%If the edge $e$ is known in the context, we will simply
%say that $\hat{a}$ is a bad (respectively, good) edge.
%%
%Similarly, we say that a path $\hat{P}$ in $\hat{T}$ is 
%{\em good} (against $e$) if all edges of $\hat{P}$ are
%good. Otherwise, we say that $\hat{P}$ is bad.
%
%We claim that we can route $\hat{r},\hat{S}_i$-flow of value at least
%half using only good paths and also avoid using $e$
%for every terminal $t$.
%%
%In particular, we prove the following lemma.

\begin{lemma}
\label{lem:flow-on-good-paths}
Let $e\in E$ and $t\in S$. Let $\hE_{bad}\subseteq \hE$ be the subset of bad edges  against $e$, and $\hE_{good}=\hE\setminus \hE_{bad}$. Consider the shallow tree $\hT' = \hT \setminus \hat{E}_{bad}$
with capacities $\{\hx'_{\he}\}_{\he\in \hE}$,
where $\hx'_{\he}=\hx_{\he}-f_{\he,e}$
for all edges $\he\in \hT$. Then $\hT'$ supports an $\hat{r},\hat{S}_t$-flow of value at least $1/2$.
\end{lemma}
\begin{proof}
%We prove the lemma using the Max-Flow-Min-Cut theorem. 
Recall that 
$\{\hf^t_{\he}\}_{\he\in \hE}$ supports an $\hat{r},{S}_t$-flow of value at 
least $2$ on $\hat{T}$.
Thus, for any cut $\hat{U}$ that separates $\hat{r}$ and $\hS_t$,
(i.e., $\hat{U}\subseteq \hV$, $\hat{r}\in\hat{U}$ and
$\hat{S}_t\subseteq (\hV\setminus \hat{U})$),
we must have
\[
\sum_{\he\in\delta^{out}_{\hT}(\hat{U})}\hf^t_{\he}\geq{2}.
\]
We will prove later the following inequality 
\begin{align}
\label{eq:flow-on-bad-edge}
 \hf_{\he}^t - f^t_{\he,e} \leq \hx_{\he} - f_{\he,e}
\mbox{ for any edge $\he \in \hE$}.
\end{align}

%We will prove Eq.~\ref{eq:flow-on-bad-edge} later.
Now we consider the capacities of edges leaving $\hat{U}$ 
in the absence of bad edges.
\begin{align*}
\sum_{\he\in\delta^{out}_{\hT}(\hat{U})\cap \hE_{good}}
      \hx'_{\he}
  &=  
\sum_{\he\in\delta^{out}_{\hT}(\hat{U})\cap \hE_{good}}
      (\hx_{\he}-f_{\he,e})=
   \sum_{\he\in\delta^{out}_{\hT}(\hat{U})}
      (\hx_{\he}-f_{\he,e})
   - \sum_{\he\in\delta^{out}_{\hT}(\hat{U})\cap \hE_{bad}}
      (\hx_{\he}-f_{\he,e})\\
   & \overset{\text{By \eqref{eq:flow-on-bad-edge}}}{\geq}     
   \sum_{\he\in\delta^{out}_{\hT}(\hat{U})}
      (\hf^t_{\he}-f^t_{\he,e})
   - \sum_{\he\in\delta^{out}_{\hT}(\hat{U})\cap \hE_{bad}}
      (\hx_{\he}-f_{\he,e})\\
      & \overset{\text{by $LP_{gst}$}}{\geq} 
   2-\sum_{\he\in\delta^{out}_{\hT}(\hat{U})}f^t_{\he,e}
   - \sum_{\he\in\delta^{out}_{\hT}(\hat{U})\cap \hE_{bad}}
      (\hx_{\he}-f_{\he,e})\\ 
      & \overset{\text{by def. bad}}{\geq} 
   2-\sum_{\he\in\delta^{out}_{\hT}(\hat{U})}f^t_{\he,e}
   - \frac{1}{2\beta}\sum_{\he\in\delta^{out}_{\hT}(\hat{U})\cap \hE_{bad}}
      f_{\he,e} \geq    2-\sum_{\he\in \hE}f^t_{\he,e}
   - \frac{1}{2\beta}\sum_{\he\in \hE}
      f_{\he,e}\\
     &   \overset{\text{by $LP_{div}$}}{\geq}    2-x_e
   - \frac{1}{2\beta}\sum_{\he\in \hE}
      f_{\he,e} \overset{\text{by $LP_{cong}$}}{\geq} 2-x_e -\frac{1}{2\beta}\cdot \beta x_e \overset{x_e\leq 1}{\geq} \frac{1}{2}    
&          
\end{align*}
Thus, by the Max-Flow-Min-Cut Theorem, 
the network $\hT'$ with capacities 
$\{\hx'_{\he}\}_{\he\in \hE}$ supports
an $\hat{r},\hat{S}_t$-flow of value at least 1/2.

It remains to prove \eqref{eq:flow-on-bad-edge}. The claim is trivially true if $f_{\he,e}=0$ since it implies $f^t_{\he,e}=0$. So, let us assume that $e$ belongs to the support of $\{f_{\he,e'}\}_{e'\in E}$. Again, we use the Max-Flow-Min-Cut Theorem.
Consider $\he=\hu\hv\in \hE$, and let $u=\psi(\hat{u})$ and
$v=\psi(\hat{v})$. 
By the constraints of $LP_{cong}$,
the graph $G$ with capacities $\{f_{\he,e'}\}_{e'\in E}$
supports a $u,v$-flow of value  $\hx_{\he}$. There must exist a minimum $u,v$-cut $U^*$ that contains the edge $e$, provided that $f_{\he,e}>0$. 
To see this, observe that $\{f_{\he,e'}\}_{e'\in E}$ induces a minimal flow network (as it is a flow itself), i.e., decreasing the capacity of any edge decreases the value of maximum flow by the same amount. So, every edge with positive capacity must contain in some minimum cut.
Consequently, we have 
\begin{align*}
\hx_{\he} - f_{\he,e} 
  =&    \left(\sum_{e'\in \delta^{out}_{G}(U^*)}f_{\he,e'}\right) - f_{\he,e} =    \sum_{e'\in (\delta^{out}_{G}(U^*)\setminus \{e\})}f_{\he,e'}\\
   \overset{\text{By $LP_{div}$}}{\geq} & \sum_{e'\in (\delta^{out}_{G}(U^*)\setminus e)}f^t_{\he,e'}
  =    \left(\sum_{e'\in \delta^{out}_{G}(U^*)}f^t_{\he,e'} \right) - f^t_{\he,e} \geq \hf^t_{\he} - f^t_{\he,e}
\end{align*}
This completes the proof.

%It remains to prove \eqref{eq:flow-on-bad-edge}.
%Again, we use the Max-Flow-Min-Cut Theorem.
%Consider the edge $\hat{a}=\hat{u}\hat{v}\in E(\hat{T})$,
%and the corresponding vertices $u=\psi(\hat{u})$ and
%$v=\psi(\hat{v})$. 
%By the constraints of the Sub-LP,
%thee graph $G$ with capacities $\{z_{\hat{a},e'}\}_{e'\in E(G)}$
%supports a $u,v$-flow of value  $\hat{x}_{\hat{a}}$.
%In fact, $\{z_{\hat{a},e'}\}_{e'\in E(G)}$ is a $u,v$-flow.
%Thus, there is a minimum $u,v$-cut $U^*$ of value exactly $\hat{x}_{\hat{a}}$.
%As we know from the constraints in the LP that
%$y^i_{\hat{a},e'} \leq z_{\hat{a},e'}$ for all edges $e'\in E(G)$,
%we have that  
%\begin{align*}
%\hat{x}_{\hat{a}} - z_{\hat{a},e} 
%  =&    \sum_{e'\in \delta^{out}_{G}(U^*)}z_{\hat{a},e'} - z_{\hat{a},e}\\
%  &=    \sum_{e'\in (\delta^{out}_{G}(U^*)\setminus e)}z_{\hat{a},e'}\\
%  &\geq \sum_{e'\in (\delta^{out}_{G}(U^*)\setminus e)}y^i_{\hat{a},e'}\\
%  &=    \sum_{e'\in \delta^{out}_{G}(U^*)}y^i_{\hat{a},e'} - y^i_{\hat{a},e}\\
%  &\geq \hat{f}^i_{\hat{a}} - y^i_{\hat{a},e},
%\end{align*}
%This completes the proof.
\end{proof}

Next consider any good path against $e$ in $\hat{T}$, say 
$\hat{P}\subseteq\hat{T}$, that connects $\hat{r}$ to a group $\hat{S}_t$.
Then $\hat{P}$ maps to an $r,t$-path in the original graph $G$ 
that does not use $e$ with probability at least $1-1/D$.  

\begin{lemma}
\label{lem:good-path-integral}
Let $e\in E$ and $\hat{P}\subseteq \hT$ be a good $\hat{r},\hat{S}_t$-path 
against $e$.
Suppose we map $\hat{P}$ to a subgraph $Q\subseteq G$ by 
sampling $(4\beta+2)\ln D$ paths from the distribution $\distP_{\he}$
for each $\he\in E(\hat{P})$. 
Then $Q\setminus\{e\}$ contains an $r,t$-path with probability at least $1-1/D$.
\end{lemma}

\begin{proof}
Consider an edge $\he=\hu\hv \in E(\hat{P})$. 
By Lemma~\ref{lem:path-distribution}, we have that 
any path $P$ sampled from $\distP_{\he}$
contains $e$ with probability
\[
\Pr_{P\sim\distP_{\he}}[e\in E(P)] 
  \leq \frac{f_{\he,e}}{\hx_{\he}}
  \overset{def. good}{\leq} 1-\frac{1}{2\beta+1}
\]
Since we sample $(4\beta+2)\ln D$ paths from $\distP_{\hat{a}}$,
the probability that all the sampled paths contain $e$ is
at most 
\[
\Pr[\mbox{all paths $P$ sampled from $\distP_{\he}$
    contain $e$}]
\leq \left(1-\frac{1}{2\beta+1}\right)^{2(2\beta+1)\ln D}
\leq \left(\frac{1}{{\bf e}}\right)^{2\ln D}
\leq \frac{1}{D^2}
\]
(here ${\bf e}$ is the base of the natural logarithm.)
We recall that $\hat{P}$ has length at most $D$.
Thus, by the union bound, 
%the probability that, for some edge $\he\in
%E(\hat{P})$, every path $P$ sampled from $\distP_{\he}$ contains $e$
%is at most 
\[
\Pr[\exists {\he\in E(\hat{P})} \mbox{ s.t. }
    \mbox{all the paths $P$ sampled from $\distP_{\he}$ 
          contain $e$}]
\leq D \cdot \frac{1}{D^2} = \frac{1}{D}.
\]
We conclude that, with probability at least $1-1/D$, for each $\he\in E(\hat{P})$ we sample at least one path in $G$ that avoids $e$: the union of such avoiding paths forms a (possibly non-simple) $r$-$t$ path that avoids $e$.
\end{proof}

Now are ready to prove the feasibility of our solution $H$
obtained from Algorithm~\ref{algo:main-algo}.

\begin{lemma}
\label{lem:output-two-connected}
The subgraph $H$ returned from Algorithm~\ref{algo:main-algo} is a feasible \tdst solution w.h.p.
\end{lemma}
\begin{proof}
By Menger's theorem, $H$ is a feasible \tdst solution iff for every edge $e\in E$ and terminal $t\in S$,  
$H\setminus \{e\}$ contains an $r,t$-path.
%We prove this lemma using a cut argument.
%Specifically, by Menger's theorem, $H$ is $2$-$(r,S)$-connected if 
%and only if for every edge $e\in E(H)$ and a terminal $t_i\in S$,  
%$H\setminus e$ contains an $r,t_i$-path.
%(It clearly holds for an edge $e\in E(G)$ that is not in $H$.)
%Let us fix our consideration to any edge $e\in E(G)$ and
%a terminal $t_i\in S$.

We claim that each subgraph $H_j\setminus \{e\}$ contains an
$r,t$-path with probability at least $1/(5D)$. First observe that, by Lemma~\ref{lem:flow-on-good-paths}, the capacities $\{\hx'_{\he}\}_{\he\in \hE}$ support
an $\hat{r},\hS_t$-flow of value at least $1/2$ through good paths. Thus, by Lemma~\ref{lem:property-GKR}, the GKR rounding algorithm guarantees that    
$\hH_j$ contains a good $\hat{r},\hat{S}_t$-path with probability at least $1/(4D)$. Given the existence of a good path in $\hH_j$, by Lemma ~\ref{lem:good-path-integral}, $H_j$ contains an $r,t$-path avoiding $e$ with probability at least $(1-1/D)$. Altogether, $H_j$ contains such a path with probability at least $(1-1/D)/(4D)\geq 1/(5D)$.

Since $H$ is a union of $20 D\ln n$ subgraphs $H_j$'s sampled
independently, the probability that no subgraphs $H_j$ contain an
$r,t$-path is at most
\[
\left(1-\frac{1}{5D}\right)^{20 D\ln n} \leq \frac{1}{n^4}.
\]
As we have at most $n$ terminals and at most $n^2$ edges,
it follows by the union bound that $H\setminus \{e\}$ contains
an $r,t$-paths, for every edge $e\in E$ and 
every terminal $t\in S$, with probability at least $1-1/n$.
%The proof of the above claim consists of two steps. First, we show that a subgraph $\hH_j$ of the shallow tree $\hT$
%contains a good (against $e$) $\hat{r},\hat{S}_t$-path  
%with probability $\Omega(1/D)$. 
%We recall that $\hat{H}_j$ is obtained by the GKR-rounding algorithm.
%By Lemma~\ref{lem:flow-on-good-paths}, we know that 
%the capacities $\{\hx_{\he}\}_{\he\in \hE}$ support
%an $\hat{r},\hS_t$-flow of value at least $1/2$ through good
%paths. 
%%
%Thus, by the property of the GKR-rounding algorithm in
%Lemma~\ref{lem:property-GKR}, $\hH_{\fab{j}}$ contains a good
%$\hat{r},\hat{S}_t$-path with probability at least $1/(4D)$. 
%
%Next, we apply Lemma~\ref{lem:good-path-integral}. 
%Given a good $\hat{r},\hat{S}_t$-path $\hat{P}$ on the shallow tree
%$\hT$, we map it to a solution subgraph of the original graph $G$
%by sampling $4\beta\log_2D$ paths from each edge $\he\in E(\hat{P})$.
%Thus, Lemma~\ref{lem:good-path-integral} implies that 
%the corresponding edges will form an $r,t$-path in \fab{$H\setminus \{e\}$}
%with probability at least $1-1/D$. 
%%
%Thus, $H_j\setminus \{e\}$ contains an $r,t$-path 
%with probability at least $1/(5D)$.
%This completes the proof.
\end{proof}

\section{Conclusions}
\label{sec:conclusions}

We presented a non-trivial approximation algorithm for \tdst on general graphs. 
Our approach crucially relies on a decomposition of a
feasible solution into two divergent Steiner trees.
It is known that an analogous decomposition does not exist for 
connectivity $k\geq 3$ \cite{Huck95,BK11}. 
However, weaker decomposition theorems would be sufficient to exploit our basic approach. 
For example, is it possible to decompose a feasible solution to
\kdst into a collection of $f(k)\cdot\polylog(n,h)$ trees so that, 
for each terminal $t\in S$ and for any edge-cut $F$ of size $k-1$, 
there exists some tree in the collection that connects $r$ to $t$
using no edges from $F$?
Such a result could be combined with our LP-rounding technique to achieve similar approximation ratios for any constant $k$. We are not aware of any such result nor of any counter-example.
To support our conjecture, we show in Appendix~\ref{sec:application-kgst} the existence of a weaker decomposition in undirected graphs that supports connectivity $\lceil k/2 \rceil$. Such decomposition allows us to design a bi-criteria approximation algorithm for a variant of \kgst, namely \kgstx,
where all the $k$ edge-disjoint paths must end at the same vertex.

Achieving a sub-polynomial approximation for \tdst in polynomial time is another obvious open problem.
However, this has been a major open problem for decades even for \dst. On the positive side, it is likely that any future progress on \dst can be extended to \tdst via our approach.
\paragraph{Acknowledgment.}
We thank R. Ravi for suggesting the variant of \kgst.

% \newpage

\bibliographystyle{abbrv}
\bibliography{2dst}

\newpage

\appendix

\section{A Reduction from \tdss to \tdst}
\label{sec:fromDSTtoDSS}

It is known that an approximation algorithm for \kdst
implies an approximation algorithm for \kdss for both
edge and vertex connectivity versions
\cite{KR96,L15-sskcon,N12-sskcon}.
The reductions of these two cases are slightly different,
but they are based on the same technique.

\paragraph{Edge-Connectivity.}

We first consider the edge-connectivity version of \tdss and \tdst.
It is known that an $\alpha$-approximation algorithm for \kdst 
yields an approximation algorithm for \kdss with a loss of
factor two \cite{KR96}.
To be precise, let $G$ be an input graph of \kdss and
$S\subseteq V(G)$ be a set of terminals,
and let $\mathcal{A}$ be an $\alpha$-approximation algorithm for \kdst.
We form an instance of \kdst by taking an arbitrary terminal $r\in
S$ as a root vertex of \kdst and taking $S'=S\setminus \{r\}$ as a
set of terminals. 
Then we solve in-rooted-version and out-rooted-version of \kdst,
separately, and take the union of the two solutions.
Thus, every terminal $t\in S\setminus\{r\}$ has 
$k$ edge-disjoint paths to and from the root. 
It then follows by the transitivity of edge-connectivity
that there are $k$ edge-disjoint paths joining every pair
of terminals. 
Therefore, this gives a $2\alpha$-approximation algorithm
for \kdss.

\paragraph{Vertex-Connectivity.}

Now, we consider the case of vertex-connectivity of \kdss 
and \kdst.
The reduction is more involved than the case of edge-connectivity
since vertex-connectivity does not have the transitivity property.
Here we need to pay an extra factor of $k^2$.
In particular, as shown in \cite{KR96,L15-sskcon,N12-sskcon},
an $\alpha$-approximation algorithm for \kdst implies
$(2\alpha+k^2)$-approximation algorithm for \kdss.

The reduction is as follows.
Let $G$ be an input graph of \kdss and
$S\subseteq V(G)$ be a set of terminals,
and let $\mathcal{A}$ be an $\alpha$-approximation algorithm for \kdst.
First, we take any subset $R$ of $k$ vertices from $S$.
The we apply any (efficient) min-cost $k$-flow algorithms on every
pair of vertices in $R$ and obtain an set of edges $E'$.
We then form an instance of the vertex-connectivity version of 
\kdst by adding an auxiliary vertex $r$ as a root and joining
it to every vertex of $R$, and then
taking $S'=S\setminus R$ as a set of terminals.
We apply the algorithm for \kdst for both in-version and
out-version and then take the union of these solutions with $E'$
(that we obtained from the min-cost $k$-flow algorithm).
It is not hard to see that the cost of the solution is at most
$(\alpha+k^2)$ times the optimum, and the feasibility can be verified using a
cut-based argument. (See \cite{L15-sskcon} for more details.)
%

%%%%%%%%%%%%%%%%%%%%%%%%%%%%%%%
%           k-GST             %
%%%%%%%%%%%%%%%%%%%%%%%%%%%%%%%

\section{Bicriteria Approximation Algorithm for a Variant of \kgst}
\label{sec:application-kgst}

In this section, we present an application of our algorithm for \tdst\ to a variant of \kgst proposed by Gupta~et~al.~\cite{GKR10}. Recall that in \kgst we wish to find a min-cost subgraph 
$H$ of a weighted undirected graph $G=(V,E)$ that has $k$ edge-disjoint paths from a given root $r$ to each \emph{group} $S_t\subseteq V$, 
$t=1,2,\ldots,h$.
In the mentioned variant, we require that all such paths end at the same vertex $s_t\in S_t$. We refer to this problem as \kgstx. Gupta~et~al. presented an $O(\log^3n\log^h\log\log n)$ approximation for the case of \tgstx, and
the algorithm of Chalermsook~et~al.\cite{CGL15} gives a bicriteria approximation algorithm\footnote{In \cite{CGL15}, the authors considered the standard version of \kgst, but the algorithm also works for the variant of \kgstx.} that provides connectivity $\Omega(k/\log n)$. We present an alternative bicriteria approximation
algorithm that outputs a solution with cost
at most $O(k\cdot D^3\log D\cdot h^{2/D}\cdot \log n)$ times the optimum
and provides connectivity at least $\lceil k/2 \rceil$.
%(Note that we may assume that $k$ is a factor of $2$;
%otherwise, we subtract the requirement by one.)

Our algorithm is now based on decomposing the optimal solution
to \kgstx into a collection of k trees $T_1,\ldots,T_{k}$ such that
for any set of edges $F$ of size $\lceil k/2\rceil-1$ in $G$,
there exists a tree $T_{i}$ that contains no edge in $F$.
The algorithm and analysis then follow along the same line
as that for \tdst.
But, we need to run the outer loop of the rounding procedure 
(Step~\ref{algo:tdst:cut-avoiding} of Algorithm~\ref{algo:main-algo})
for $O(kD\ln{n})$ times instead of $O(D\ln{n})$
because the number of edge-cuts that we have to consider is now $n^{O(k)}$. 
This incurs an extra factor of $O(k)$ in the approximation
guarantee.

It remains to show that the above decomposition exists.
Observe that an optimal solution $H$ to \kgst forms
a graph that is $k$-edge-connected on the set
$S^*=\{r,s_1,\ldots,s_h\}$.
Thus, if we replace each undirected edge $\{u,v\}$ by
two directed edges $uv$ and $vu$, then we have a
directed graph $\hat{H}$ such that $S^*$ 
is $k$ (strongly) edge-connected on $\hat{H}$.
We may apply a splitting-off theorem to get rid of all the
Steiner vertices (vertices in $V(H)\setminus S^*$),
resulting in a directed $k$-edge-connected graph $\hat{H}'$
whose vertex-set is $S^*$.
Then, by Edmonds' Disjoint Arborescence Packing Theorem
\cite{Edmonds73},
we know that $\hat{H}'$ contains $k$ edge-disjoint (out) spanning 
arborescences $\hat{T}'_1,\ldots,\hat{T}'_k$. 
Mapping them back to the original graph $H$,
we have a collection of $k$ trees 
$\mathbb{T}=\{T_1,\ldots,T_k\}$ such that any undirected edge $\{u,v\}$ 
is contained in at most two trees in $\mathbb{T}$.
So, for any set of $\lceil k/2\rceil-1$ edges $F\subseteq E(H)$, there must exist a
tree $T_i\in\mathbb{T}$ that contains no edge of $F$
and connects every terminal to the root.
Therefore, we have the decomposition as claimed.

Note that this is an evidence that a weaker version of 
the decomposition theorem (Theorem~\ref{thm:independent-tree})
might exist.
The decomposition implies the following theorem.

\begin{theorem}[Bicriteria \kgstx]
\label{thm:bicriteria-kdstx}
For any $D\in[\log_2{h}]$, there exists a randomized approximation
algorithm that runs in $n^{O(D)}$ time
and outputs a feasible solution $H$ to \kgstx
with cost $O(k\cdot D^3\log D\cdot h^{2/D}\cdot \log n)$ 
times the optimum
and provides connectivity at least $\lceil  k/2 \rceil$.
\end{theorem}

\medskip

\noindent{\bf Remark.}
For the case of \tgst, our algorithm gives a ``true'' approximation
algorithm for both \tgst and \tgstx.
To see this, we split each (undirected) edge $\{u,v\}$ of the input
graph into two directed edges $uv$ and $vu$ with the same cost.
We then add a terminal $s_{t}$ for each group $\hS_t$ and joining each
vertex $v\in \hS_t$ to $s_t$ by a (directed) edge $vs_t$ with
zero-cost. 
This reduces \tdst to \tgst, but we have a small issue
that the composition in Theorem~\ref{thm:independent-tree}
may give trees $T_1$ and $T_2$ such that the corresponding 
two edge-disjoint $r,s_t$-paths $P_1$ and $P_2$ contain both
$uv$ and $vu$ edges. 
However, we may use a stronger form of Theorem~\ref{thm:independent-tree} where we additionally require
that the paths $P_1$ and $P_2$ are \emph{strongly} divergent, i.e.,
only one of $uv$ and $uv$ can be contained in $E(P_1)\cup E(P_2)$
\cite{GeorgiadisT16,Kovacs07}.
Our approximation guarantee matches the results in \cite{GKR10}
(albeit, with worse running time).

\end{document}